\documentclass[envcountsame]{llncs}

\usepackage[utf8]{inputenc}
\usepackage[T1]{fontenc}
\usepackage{amsfonts}
\usepackage{amssymb}
\usepackage{amsmath}
\usepackage{latexsym}
\usepackage{stmaryrd}
\usepackage{wasysym}
\usepackage{graphicx}
\usepackage{rotating}
\usepackage{pifont}
\usepackage{booktabs}
\usepackage{array}
\usepackage{multirow}
\usepackage{tikz}
\usepackage{color}
\usepackage{xcolor}
\usepackage{enumitem}
\usepackage{xspace}
\usepackage{url}
\usepackage{cite}
\usepackage[obeyDraft]{todonotes}
\usepackage{mathtools}
\usepackage{thm-restate}

\newcommand{\Fmc}{\ensuremath{\mathcal{F}}\xspace}
\newcommand{\Imc}{\ensuremath{\mathfrak{I}}\xspace}

\newcommand{\Pmc}{\ensuremath{\mathcal{P}}\xspace}

\newcommand{\nats}{\ensuremath{\mathbb{N}}\xspace}
\newcommand{\reals}{\ensuremath{\mathbb{R}}\xspace}

\DeclareFontEncoding{LS1}{}{}
\DeclareFontSubstitution{LS1}{stix}{m}{n}
\DeclareSymbolFont{ssymbols4}      {LS1}{stixbb}   {m}{it}
\DeclareMathSymbol{\diamondleftblack}        {\mathord}{ssymbols4}{"D1}

\newcommand{\ptdl}{\textsl{TLD-Lite}\xspace}

\newcommand{\dd}{\ensuremath{\raisebox{0.37ex}{\scalebox{0.74}{$\diamondleftblack$}}}\xspace}

\newcommand{\DL}{\textsl{DL-Lite}}

\newcommand{\coNExpTime}{\textsc{coNExpTime}}
\newcommand{\PSpace}{\textsc{PSpace}}
\newcommand{\nxt}{{\ensuremath\raisebox{0.25ex}{\text{\scriptsize$\bigcirc$}}}}

\newcommand{\bx}{{\ensuremath\raisebox{0.1ex}{$\Box$}}}

\newcommand{\Since}{\mathbin{\mathcal{S}}}
\newcommand{\U}{\mathbin{\,\mathcal{U}}\xspace}

\newcommand{\I}{\mathcal{I}}
\newcommand{\A}{\mathcal{A}}
\newcommand{\TO}{\mathcal{O}}
\newcommand{\T}{\mathcal{T}}
\newcommand{\R}{\mathcal{R}}

\newcommand{\K}{\mathcal{K}}

\newcommand{\ind}{\mathsf{ind}}

\newcommand{\avec}[1]{\boldsymbol{#1}}

\newcommand{\dom}{\mathop{\textit{dom}}}

\newcommand{\av}{\boldsymbol a}

\newcommand{\J}{\mathcal{J}}

\newcommand{\SVbox}{\mathop{\ooalign{$\Box$ \cr \kern0.4ex
    \raisebox{0.25ex}{\scalebox{0.7}{$*$}}}\rule{0pt}{1.5ex} \kern-0.7ex}}
\newcommand{\SVdiamond}{\mathop{\ooalign{\scalebox{1.1}{$\Diamond$} \cr \kern0.35ex
    \raisebox{0.25ex}{\scalebox{0.7}{$*$}}} \kern-0.9ex}}

\newcommand{\diamondplus}{\mathop{\ooalign{\scalebox{1.1}{$\Diamond$} \cr \kern0ex
    \raisebox{0.25ex}{\scalebox{0.9}{$+$}}} \kern-0.8ex}}
\newcommand{\diamondminus}{\mathop{\ooalign{\scalebox{1.1}{$\Diamond$} \cr \kern0ex
    \raisebox{0.25ex}{\scalebox{0.9}{$-$}}} \kern-0.8ex}}

\newcommand{\Geom}{\ensuremath{\mathsf{Geom}}\xspace}

\begin{document}

\title{Cutting Diamonds}
\subtitle{Temporal DLs with Probabilistic Distributions over Data}
%
\author{Alisa Kovtunova \and Rafael Pe\~naloza}

\institute{KRDB Research Centre, Free University of Bozen-Bolzano, Italy. \email{\{alisa.kovtunova,rafael.penaloza\}@unibz.it}}

\maketitle
\begin{abstract}
Recent work has studied a probabilistic extension of the temporal logic LTL that refines the eventuality (or diamond) constructor with a probability distribution on when will this eventuality be satisfied. In this paper, we adapt this notion to a well established temporal extension of DL-Lite, allowing the new probabilistic constructor only in the ABox assertions. We investigate the satisfiability problem of this new temporal DL over equiparametric geometric distributions. 

%
\end{abstract}

\section{Introduction}
Combinations of DLs with temporal formalisms have been widely investigated since the early work of 
\cite{DBLP:conf/aaai/Schmiedel90}; we refer the reader 
to~\cite{DBLP:books/cu/Demri2016,DBLP:conf/time/LutzWZ08,Gabbay2003-GABMML,DBLP:conf/time/ArtaleKKRWZ17} 
for detailed surveys of the area. Despite the different visions of the problem presented, logical theories that encode a 
domain of interest are always represented by factual statements. However, speaking about the future by itself can imply 
(probabilistic) uncertainty. 

For example, insurance companies estimate the probability of the insured event in the duration period of the policy 
based in a number of factors; e.g., for a life insurance, they consider health condition, number of children, habits, 
sun radiation in home region, etc. This defines the extent of monthly payment for a customer. If one uses a 
classical temporal DL, one can only express that everyone dies 
($\textsf{LivingBeing} \sqsubseteq \Diamond \textsf{Dead}$), and miss the golden goose of insurers.
 
There is also a large pool of proposals for probabilistic DLs, e.g., \cite{DBLP:journals/semweb/RiguzziBLZ15,DBLP:conf/dlog/SazonauS15,DBLP:journals/jair/Gutierrez-Basulto17,DBLP:conf/sum/PenalozaP17}, that differ widely in many fundamental aspects, like the 
way in which probabilities are used, in the syntax, in the chosen semantics, and in the possible application. We refer to~\cite{klinov} for a (now slightly outdated) survey. 
There 
are two main views of probability~\cite{DBLP:journals/ai/Halpern90}, statistical and subjective. While the statistical 
view considers a probability distribution over a domain that specifies the probability for an individual in the domain to 
be randomly picked, we choose the subjective view, which specifies the probability distribution over a set of possible 
worlds.
In our case, a world would be a possible evolution of a system; that is, a standard temporal DL interpretation. 
The authors of~\cite{DBLP:journals/jair/Gutierrez-Basulto17} argue the subjective semantics provides an appropriate 
modelling for probabilistic statements about individuals, e.g., the statement ``there's at least 80\% chance of not having an earthquake tomorrow'' implies that an earthquake will either occur or will not occur. Thus, in the set of possible worlds, there 
are some structures in which an earthquake happens and others in which it does not. This type of uncertainty can also be called epistemic, because it regards probabilities as the degree of our belief.

This paper presents the language~\ptdl, which, to the best of our knowledge, is a first probabilistic extension of temporal description logics 
and aims at closing the gap between probabilistic and temporal extensions of lightweight DLs.
We propose a new view on one of the pillars of the linear temporal logic LTL, the eventuality (or 
diamond) constructor, which expresses that some property will hold at some point in the 
future. The specification of the diamond operator is nondeterministic and, thus, can be too rough; indeed, with the basis of actual 
life experience one often has an
idea---albeit uncertain---about \emph{when} the property may hold. Employing basic notions from statistical 
analyses, we combine an ontology of a well established temporal extension of DL-Lite proposed 
in~\cite{DBLP:journals/tocl/ArtaleKRZ14} with temporal data specified by the geometric distribution with parameter $p \geq \frac{1}{2}$. 

This logic allows us to reason over time dimension with uncertain knowledge. For example, if one builds an 
earthquake-resistant house, one wants to be sure an earthquake does not ruin it before building bricks and concrete blocks are properly reinforced. Also, if 
local hospitals take vacation on Sundays, an earthquake in the construction camp this day of the week inflicts 
heavier losses.

For the resulting temporal lightweight description logic with distributions, aka \ptdl, we present 
the formal semantics underlying the language, introduce the probabilistic formalism realised by means of the probabilistic constraint $\dd_\delta$, called the \emph{distribution
eventuality}, and investigate the satisfiability 
problem. Consistency can be checked by a deterministic algorithm using exponential space in the size of input. An open question is whether the result can be improved to match the \coNExpTime{} upper bound obtained in~\cite{ourKR} for the propositional temporal formula 
with only one instance of the distribution eventuality $\dd_\delta$. 
This new refined diamond constructor includes a 
discrete probability distribution $\delta$ that can be used to specify the likelihood of observing the property of interest, for the 
first time, at each possible point in time.

\section{Preliminaries}

We briefly introduce the basics of probability theory and temporal extensions of description logics.

\subsection{Probability Theory}
We start by providing the basic notions of probability needed for this paper. For a deeper study on probabilities,
we refer the interested reader to~\cite{Bill-95}.
Let $\Omega$ be a set called the \emph{sample space}. A \emph{$\sigma$-algebra} over $\Omega$ is a class 
\Fmc of subsets of $\Omega$ that contains the empty set, and is closed under complements and under 
countable unions. A \emph{probability measure} is a function $\mu:\Fmc\to [0,1]$ such that $\mu(\Omega)=1$, 
and for any countable collection of pairwise disjoint sets $E_i\in\Fmc$, $i\ge 1$, it holds that
$
\mu(\bigcup_{i=1}^\infty E_i)=\sum_{i=1}^\infty \mu(E_i).
$
The probability of a set $E\in\Fmc$ is $\int_{\omega\in E}\mu d\omega$, where the integration is made w.r.t.\
the measure $\mu$. 

A usual case is when $\Omega$ is the set $\reals$ or all real numbers, and \Fmc is the standard Borel 
$\sigma$-algebra over \reals; that is, the smallest $\sigma$-algebra containing all open intervals in \reals.
In this case, $\mu$ is called a \emph{continuous} probability measure, and the integration defining the probability
of a set $E$ corresponds to the standard Riemann integration.

If $\Omega$ is a countable (or finite) set, the standard $\sigma$-algebra is formed by the power set of 
$\Omega$, and a probability measure $\mu$ is uniquely determined by a function $\mu:\Omega\to[0,1]$. Given
a set $E\subseteq\Omega$, 
$\mu(E)=\sum_{\omega\in E} \mu(\omega)$; that is, the probability of a set is the sum of the probabilities of the
elements it contains. In this case, $\mu$ is called a \emph{discrete} probability measure. In addition, if $\mu(\omega)>0$ for all $\omega\in\Omega$, then $\mu$ is \emph{complete}. In contrast to our definition, in a classical Kolmogorov probability space the completeness of $\mu$ only requires it to be necessarily defined for every $\omega\in\Omega$.    
When $\Omega$ is the set of all natural numbers \nats, we specify the distribution $\mu$ as a function
$\mu:\nats\to[0,1]$. 

A simple example of a complete discrete distribution is the \emph{geometric} distribution. 
The geometric distribution $\Geom(p)$ with a parameter (the probability of 
success) $p\in(0,1)$ is defined, for every $i\in\nats$, by $\mu(i)=(1-p)^{i-1}p$. This distribution describes the 
probability of 
observing the first success in a repeated trial of an experiment at time $i$. 
Returning back to the insurance example, according to~\cite{Slud2006ActuarialMA}, an attained integer age at death has the geometrical distribution if the force of mortality were constant at all ages.

\subsection{Temporal DLs}
Temporal \DL{} logics are extensions of standard \DL{} description logics introduced 
by~\cite{DBLP:journals/jar/CalvaneseGLLR07,DBLP:journals/jair/ArtaleCKZ09}. Similarly 
to~\cite{DBLP:conf/time/ArtaleKLWZ07}, since we want to reason about the future, we also allow applications of 
the discrete unary future operators $\nxt$ (``in the next time point''), $\bx$ (``always in the future''), $\Diamond$ (``eventually in the future'') and the binary operator 
$\U$ (``until'') to basic concepts. We use the non-strict semantics for $\U$, $\Diamond$ and 
$\bx$ in the sense that their semantics includes the current moment of time. 

Formally,  \ptdl{} contains \emph{individual names} $a_0,a_1,\dots$, \emph{concept names} 
$A_0,A_1,\dots$, \emph{flexible role names} $P_0,P_1,\dots$ and \emph{rigid role names} $G_0,G_1,\dots$.  
\emph{Roles} $R$, \emph{basic concepts} $B$ and \emph{concepts} are defined by the grammar
\begin{eqnarray*}
R \ \ ::= \ \  S \ \ \mid \ \ S^-, \qquad  S\ \ ::= \ \  P_i \ \ \mid \ \ G_i,\qquad B \ \ ::=\ \ 
\bot \ \ \mid\ \
A_i \ \ \mid\ \
\exists R, \\
C \ \ ::=\ \ B  \ \ \mid\ \ \neg C  \ \ \mid\ \ \nxt C  \ \ \mid\ \ \Diamond C  \ \ \mid\ \ \bx C  \ \ \mid\ \ C \U C  \ \ \mid\ \ C\sqcap C.
\end{eqnarray*}
%
We denote the nesting of temporal operators as the superscription of the temporal operator from the set 
$\{\bx, \nxt\}$; that is, $\bx^0A=\nxt^0A=A$, and $\bx^{n+1}A=\bx\bx^nA$, $\nxt^{n+1}A=\nxt\nxt^nA$.

%
A temporal \emph{concept inclusion} (CI) takes the form
%
$C_1 \sqsubseteq C_{2}$,
%
while temporal \emph{role inclusions} (RI) are of the form
%
$R_1  \sqsubseteq R_{2}$.
As usual $C_1 \equiv C_2$ abbreviates $C_1 \sqsubseteq C_2$ and $C_2 \sqsubseteq C_1$. 
%
All CIs and RIs are assumed to hold globally (over the whole timeline). 
Note that an CI can express that a basic concept is \emph{rigid}, i.e., interpreted in the same way at every point of time.
%
%
%
%
%
A temporal \emph{TBox} $\T$ (resp. \emph{RBox} $\R$) is a finite
set of temporal CIs (resp., RIs). Their union $\TO= \T\cup \R$ is called a temporal \emph{ontology}. 
%
Since the non-strict operators are obviously definable in terms of the strict ones, which do not include the current moment, temporal CIs and RIs are expressible in terms of \PSpace-complete $T_{\U\Since}\DL_\textit{bool}^{\smash{\mathcal{N}}}$ logic~\cite{DBLP:journals/tocl/ArtaleKRZ14}. 

A \emph{temporal interpretation}
is a pair $\I =(\Delta^\I,\cdot^{\I(n)})$, where $\Delta^\I \ne\emptyset$ and
$$\I(n)=(\Delta^\I, a_0^\I,\dots ,A_0^{\smash{\I(n)}}, \dots,P_0^{\smash{\I(n)}},\dots , G_0^{\smash{\I}},\dots ,)$$
contains a standard DL interpretation for each time instant $n \in \nats$ of the ordered set \mbox{$(\nats,<)$}, that
is, $a_i^{\I}\in \Delta^\I$, $A_i^{\smash{\I(n)}}\subseteq \Delta^\I$
and $P_i^{\smash{\I(n)}},G_i^{\smash{\I}}\subseteq\Delta^\I\times\Delta^\I$. The domain $\Delta^\I$ and the interpretations $a_i^\I\in \Delta^\I$ of the individual names and $G_i^{\smash{\I}}\subseteq\Delta^\I\times\Delta^\I$ of rigid role names are the same for all
$n\in \nats$, thus, we adopt the \emph{constant domain assumption}. However, we do not assume the unique name, since neither functionality nor number restrictions are applied to this logic.

The DL and temporal constructs are interpreted in $\I(n)$ as follows:  
\begin{align*} 
(R^-)^{\I(n)} & {} = \{ (x,y) \mid (y,x) \in R^{\I(n)} \},\\
 (\exists R)^{\I(n)} & {} = \bigl\{ x \mid  (x,y) \in R^{\I(n)}, \text{ for some } y \bigr\},\\
 (\neg C)^{\I(n)} & {} = \Delta^\I \setminus C^{\I(n)},\\
 (C_1 \sqcap C_2)^{\I(n)} & {} = C_1^{\I(n)} \cap C_2^{\I(n)},\\
 (\Diamond C)^{\I(n)} & {} = \bigcup_{k\geq n} C^{\I(k)},\\
 (\Box C)^{\I(n)} & {} = \bigcap\nolimits_{k \geq n} C^{\I(k)},\\
 (\nxt C)^{\I(n)} & {} = C^{\I(n+1)}, \\
 (C_1\U C_2)^{\I(n)} & {} = \bigcup_{k\geq n}\left( C_2^{\I(k)}\cap \bigcap\nolimits_{k>l\geq n} C_1^{\I(l)}\right). 
\end{align*}
As usual, $\bot$ is interpreted by $\emptyset$ and $\top$ by $\Delta^{\smash{\I}}$ for concepts and by 
$\Delta^{\smash{\I}}\times\Delta^{\smash{\I}}$ for roles. 
As mentioned before, CIs and RIs are interpreted in $\I$ \emph{globally} in the sense that they hold in $\I$ if 
$C_1^{\smash{\I(n)}} \subseteq C_{2}^{\smash{\I(n)}}$ and $R_1^{\smash{\I(n)}}\subseteq R_2^{\smash{\I(n)}}$
hold  for all $n \in \nats$.
%
Given an inclusion $\alpha$ and a temporal interpretation $\I$, we write $\I\models\alpha$ if $\alpha$ holds in $\I$.

\subsection{Distributions of Data Instances over Time} 

In temporal variants of \DL{}, instances from the ABox can be associated with temporal constructors as well.
In our logic \ptdl{}, in addition to the standard constructors used also in the ontology, we allow a probabilistic
constructor that provides a distribution of the time needed until the assertion is observed. Formally,
a \ptdl{} \emph{ABox} (or \emph{data instance}) is a finite set $\A$ of atoms of the form
\begin{align*}
\nxt^{n} A(a), \qquad &&\nxt^{n} \dd_\delta A(a), && \nxt^{n} \neg A(a), \\
\nxt^{n} R(a,b), \qquad &&\nxt^{n} \dd_\delta R(a,b) && \nxt^{n} \neg R(a,b),
\end{align*}
where $a$, $b$ are individual names,
$n \in \nats$ and $\delta$ is a complete distribution over \nats. 
The new constructor $\dd_\delta\theta(\av)$, for $\theta=A$, $\av=a$ and $\theta=R$, $\av=(a,b)$, expresses that 
the time until the event $\theta(\av)$ is first observed has distribution $\delta$. 
We denote by $\ind(\A)$ the set of individual names in
$\A$. 
A \ptdl{} \emph{knowledge base} (KB) $\K$ is a pair $(\TO,\A)$,
where $\TO$ is a temporal ontology and $\A$ a \ptdl{} ABox.

To render the probabilistic properties, the semantics of \ptdl{} is based on the \emph{multiple-world} approach. 
A \ptdl interpretation is a pair $\Pmc=(\Imc,\mu)$, where $\Imc$ is a set of temporal interpretations $\I$ 
and $\mu$ is
a probability distribution over $\Imc$. 
Given a set of temporal interpretations $\Imc$, a concept name or a role $\theta$, individual names $\av$ and 
$n\in\nats$, let $\Imc_{\av,n}^\theta:=\{\I\in\Imc\mid \av^{\I}\in \theta^{\I(n)}\}$. For $\Pmc=(\Imc,\mu)$, we interpret the probabilistic construct in $\I\in \Imc$ at the time point $n$ over the set of individual names as 
\begin{align} 
(\dd_\delta\theta)^{\I(n)} = \{ \av^\I \; {\mid} 
\;\, \mu(\Imc_{\av,n+i}^\theta{\setminus}\bigcup_{j=0}^{i-1}\Imc_{\av,n+j}^\theta)=\delta(i) 
	\text{ for all }i{\ge} 0 \}. \label{eq:probab}
\end{align}
In contrast to~\cite{ourKR}, we do not require the unique constant domain for all interpretations in the set $\Imc$. 

$\Pmc=(\Imc,\mu)$ is a \emph{model} of $\K=(\TO,\A)$ and write
$\Pmc\models \K$ if, for any $\I\in \Imc$, 
\begin{itemize}
\item all concept and role inclusions from $\TO$ hold in $\I$, i.e. $\I\models \alpha$ for all $\alpha\in \TO$;
\item $a^{\smash{\I}}\in A^{\smash{\I(n)}}$ for $\nxt^n A(a)\in \A$, and $(a^{\smash{\I}},b^{\smash{\I}})\in R^{\smash{\I(n)}}$ for $\nxt^n R(a,b)\in \A$;\\
$a^{\smash{\I}}\not\in A^{\smash{\I(n)}}$ for $\nxt^n \neg A(a)\in \A$, and $(a^{\smash{\I}},b^{\smash{\I}}) \not\in R^{\smash{\I(n)}}$ for $\nxt^n \neg R(a,b)\in \A$;
\item $\av^{\smash{\I}} \in (\dd_\delta \theta)^{\smash{\I(n)}}$, for all $\nxt^n\dd_\delta \theta(\av)\in \A$. 
\end{itemize}
Similarly to the standard case, a KB $\K=(\TO,\A)$ is \emph{consistent} if it has a model.
As it is obvious from our semantics, we are using the standard open-world assumption from DLs.

There are several reasons why we allow the probabilistic operator only in ABox instances: (i)~semantic: each 
model of $\K$ can differ in the anonymous part, while the definition~\eqref{eq:probab} requires common object 
names for all temporal interpretations in $\Imc$; (ii)~computational: \DL{} allows infinitely many anonymous objects; 
bounding the probability to the ABox objects ensures existence of a model with a finite number of terms, i.e., concept 
names or roles, prefixed with $\dd_\delta$; (iii)~even leaving out the anonymous part of \ptdl{}, CIs and RIs can also 
express infinitely many times repeated events for ABox objects, e.g., $C \sqsubseteq \nxt^2 C$.   

\begin{remark}
The restrictions in the syntax for avoiding uncountable models are mainly of a technical nature: describing the distribution $\mu$ over an uncountable set of temporal interpretations needs measure-theoretic notions; and verifying the existence of uncountable models requires more advanced machinery. 
\end{remark}
The interpretation $\Pmc=(\Imc,\mu)$ is \emph{countable} if the set $\Imc$ contains countably many temporal 
interpretations. In~\cite{ourKR} it was shown that the combination $\Box\dd_\delta$ can only be satisfied by
an uncountable interpretation. \ptdl KBs allow the constructor \dd only in ABox instances, and these occurrences need only a finite number of time points to be satisfied. Hence, \ptdl{} has the countable-model property.

\begin{theorem}\label{th:countability}
If a \ptdl{} KB $\K$ is satisfiable, it has a countable model.
\end{theorem}
\begin{proof}
%
Let $\Pmc=(\Imc,\mu)$ be a model of $\K$, and $\avec{d}$ be the number of all \dd-data instances appearing in 
$\A$. If $\avec{d}=0$, then $\K$ does not contain any \ptdl instances, and, for every $\I\in\Imc$, an interpretation $\Pmc_\I=(\{\I\},\mu_{\I})$, where $\mu_{\I}(\{\I\})=1$, is a model of $\K$.
Otherwise, by semantics, for every instance $\nxt^{n_i}\dd_{\delta_i}\theta_i(\av_i)\in \A$, where $1\leq i\leq \avec{d}$, we have 
\begin{align}\label{eq:joint-delta}
\mu(\Imc_{\av_{i},n_{i}+k}^{\theta_i}\setminus\bigcup_{j=0}^{k-1}\Imc_{\av_i,n_i+j}^{\theta_i})=\delta_i(k),
\end{align}
for any $k\geq 0$. Since the domain of probability functions, as a $\sigma$-algebra, is closed under countable \emph{intersections}, the joint function $\mu(\J^{k_1,\dots,k_{\avec{d}}}_{\A,\Imc})$ is also defined for the \ptdl model $\Pmc=(\Imc,\mu)$, where 
\begin{align}
\J^{k_1,\dots,k_{\avec{d}}}_{\A,\Imc}=\bigcap_{i=1}^{\avec{d}} \left( \Imc_{\av_{i},n_{i}+k_{i}}^{\theta_i}\setminus\bigcup_{j=0}^{k_{i}-1}\Imc_{\av_i,n_i+j}^{\theta_i} \right). \label{eq:frac-interpretation}
\end{align}
Note that $\sum_{\{k_1,\dots,k_{\avec{d}}\}\in \nats^{\avec{d}}} \mu(\J^{k_1,\dots,k_{\avec{d}}}_{\A,\Imc})=1$.

Now from a (possibly) uncountable $\Pmc$ we build a new countable \ptdl interpretation $\Pmc'=(\Imc',\nu)$ by 
assigning an appropriate weight to a representative interpretation $\I$ of a set 
$\J^{k_1,\dots,k_{\avec{d}}}_{\A,\Imc}$ for every $k_1,\dots,k_{\avec{d}}\geq 0$.

Initially we assume $\Imc'=\emptyset$. For all $k_1,\dots,k_{\avec{d}} \geq 0$, we consider a subset 
$\J_{\A,\Imc}^{k_1,\dots,k_{\avec{d}}}\subseteq\Imc$ defined by~\eqref{eq:frac-interpretation}. If 
$\J_{\A,\Imc}^{k_1,\dots,k_{\avec{d}}}=\emptyset$ and $\mu(\J^{k_1,\dots,k_{\avec{d}}}_{\A,\Imc})=0$, then we 
assign $\nu(\J^{k_1,\dots,k_{\avec{d}}}_{\A,\Imc'})=\nu(\emptyset)=0$. Otherwise, we pick any interpretation 
$\I\in \J^{k_1,\dots,k_{\avec{d}}}_{\A,\Imc}$ as a representative and set $\Imc'=\Imc'\cup \I$ with 
$\nu(\J^{k_1,\dots,k_{\avec{d}}}_{\A,\Imc'})=\nu(\I)=\mu(\J^{k_1,\dots,k_{\avec{d}}}_{\A,\Imc})$. In the general case, 
the last equality, $\nu(\I)$, can be equal to $0$.

As $\avec{d}$ is finite and at each step of the procedure we add at most one interpretation, $\Pmc'$ is countable. In 
order to show $\Pmc'\models \K$, we notice that, for any axiom $\alpha\in\TO$, we have $\I\models \alpha$, for any 
$\I\in \Imc$. Since $\Imc'\subseteq\Imc$, we have the statement, $\Pmc'\models \alpha$. A similar argument can be 
applied to the \dd-free ABox assertions.

Consider an instance $\nxt^n\dd_\delta \theta(\av)\in \A$. By construction of $\Pmc'=(\Imc',\nu)$, for any $k\geq 0$, 
it holds that
\begin{align*}
\nu(\{\I\in\Imc'{\mid} \av^{\I}\in \theta^{\I(n+k)}\}{\setminus}\bigcup_{j=0}^{k-1}\{\I\in\Imc'{\mid} \av^{\I}\in \theta^{\I(n+j)}\}) = {} 
&\sum_{\mathclap{\{{k_1,\dots,k_{\avec{d}}\}\setminus \{k\}}\in\nats^{\avec{d}-1}}} \nu(\J^{k_1,\dots,k_{\avec{d}}}_{\A,\Imc'}) \\ = {} & 
\sum_{\mathclap{\{{k_1,\dots,k_{\avec{d}}\}\setminus \{k\}}\in\nats^{\avec{d}-1}}} \mu(\J^{k_1,\dots,k_{\avec{d}}}_{\A,\Imc})=  
\delta(k).
\end{align*} 
By semantics, $\Pmc'\models\nxt^n\dd_\delta\theta(\av)$. Thus, the \ptdl interpretation $\Pmc'=(\Imc',\nu)$ is a countable model of $\K$.
%
\qed
\end{proof}

\section{Deciding Satisfiability}
We now focus on the problem of deciding whether a given \ptdl KB $\K$ is satisfiable.
%
The semantics of the \dd-operator given by~\eqref{eq:probab} in a combination with a nondeterministic temporal operator can give interesting results. 
\begin{example}
Consider the KB $\K=(\{B \equiv \Diamond A\}, \{\dd_\delta B(a)\})$ for a complete distribution $\delta$. $\K$ is
unsatisfiable, since, for any temporal interpretation $\I$, the statement $\I,1\models\Diamond A$ implies 
$\I,0\models \Diamond A$, which contradicts the semantics of \dd-operator~\eqref{eq:probab} with $i=1$, by which, for any model $(\Imc,\mu)$ of the KB,
$$\mu(\Imc_{a,1}^{\Diamond A}\setminus\Imc_{a,0}^{\Diamond A})=\delta(1)>0.$$
%
Namely, one should not be able to say that there is any positive probability of satisfying $\neg\Diamond A(a)$ in a 
time point $0$, and $\Diamond A(a)$ in any later time $m>0$.
\end{example}
%
%

\subsection{Multidimensional Matrix}
To represent a model $\Pmc=(\Imc,\mu)$ of a \ptdl{} KB $(\TO,\A)$ with $\avec{d}$ instances of $\dd_p$\mbox{-}data assertions in the ABox,
%
we introduce a $\avec{d}$-dimension infinite matrix $M_\Pmc$ with elements from the range of $\mu$ such that 
$M_\Pmc(k_1,\dots,k_{\avec{d}})=\mu(\J^{k_1,\dots,k_{\avec{d}}}_{\A,\Imc})$, where 
$\J^{k_1,\dots,k_{\avec{d}}}_{\A,\Imc}$ is defined by \eqref{eq:frac-interpretation}. Notice that the mapping from a 
model to a matrix is surjective: similarly to the proof of Theorem~\ref{th:countability}, the mapping merges equivalent 
(from the point of view of \dd{} unravelling) temporal interpretations together. Stepping away from exact \ptdl{} 
interpretations, by the same argument as in Theorem~\ref{th:countability}, we show the following result.
\begin{theorem}\label{th:matrix}
A \ptdl{} KB $\K$ is satisfiable iff there is a matrix $M$ of elements from the interval $[0,1]$ such that
\begin{itemize}
\item for any $\{k_1,\dots,k_{\avec{d}}\}\in \nats^{\avec{d}}$, if $M(k_1,\dots,k_{\avec{d}})>0$ then the temporal KB 
$(\TO,\A_{k_1,\dots,k_{\avec{d}}})$, where $\A_{k_1,\dots,k_{\avec{d}}}$ is the \dd-free ABox
\begin{equation}
\label{eq:matrix-formula}
\bigcup_{i=1}^{\avec{d}} \{ \bigcup_{j=0}^{{k_i}-1}  \nxt^{n_i+j} \neg\theta_i (\av_i) \cup \nxt^{n_i+k_i}\theta_i(\av_i) \} \cup \A \setminus \{\bigcup_{i=1}^{\avec{d}} \nxt^{n_i}\dd_{\delta_i} \theta_i(\av_i) \}, 
\end{equation}
is satisfiable; and
\item for any $1\leq i\leq \avec{d}$, 
\begin{equation}\label{eq:matrix-sum}
\sum_{\{{k_1,\dots,k_{\avec{d}}\}\setminus \{k_i\}}\in\nats^{\avec{d}-1}}  M(k_1,\dots,k_{\avec{d}})=\delta_i(k_i)
\end{equation}
\end{itemize}
\end{theorem} 
We consider matrix entries starting with zero; i.e., $M(0,\dots,0)$ is the first element of the matrix $M$.
%

Theorem~\ref{th:matrix} allows us to avoid providing explicitly a model for a satisfiable \ptdl{} KB, since the matrix 
ensures it existence. But it does not provide an efficient solution or even an algorithm for the satisfiability problem, since it requires an infinite matrix. However, as each non-zero element corresponds to a 
classical temporal KB, we use properties of the probabilistic distribution and establish a periodical property of matrix 
entries to bound the size of the matrix. 

\subsection{Bernoulli Processes}
So far we have introduced the \ptdl{} KB in general terms. In the following we
focus on the special case where all used distributions describe a Bernoulli process; that is, we consider only
the geometric distribution, $\Geom(p)$ with $0<p<1$; notice that this
distribution is complete. To simplify the notation, we will simply write $\dd_p$ for this distribution. 
%
The following example demonstrates the semantics of the operator $\dd_p$.

\begin{example}
Let $\delta$ be the geometric distribution $\Geom(\frac{1}{2})$.
The \ptdl{} ABox $\{\dd_{\frac{1}{2}} H(a), \nxt\dd_{\frac{1}{2}} T(a)\}$ describes two experiments observing a 
repeated flip of
the coin $a$, where $H$ means that the coin landed {\sf{heads}}, and $T\equiv \neg H$ that it landed {\sf{tails}}.
\end{example}
The geometric distribution $\Geom(p)$ has some valuable to us properties:
\begin{enumerate}
\item for any $j>i\geq 0$, we have $\Geom(p)(i) > \Geom(p)(j)$; and
\item\label{prop:second} for any $p > \frac{1}{2}$ and $i\geq 0$, $\Geom(p)(i) > \sum_{j>i} \Geom(p)(j)$. 
If $p=\frac{1}{2}$, then this inequality becomes an equality. 
\end{enumerate}
We also assume that, for all ABox instances, the parameter of the geometric distribution $p$ is unique, 
$\frac{1}{2}\leq p<1$. To simplify presentation, particularly matrix-wise, we consider only the case of $\avec{d}=2$. 
For an arbitrary finite $\avec{d}$, the reasoning we provide below is the same, but requires the use of more 
cumbersome notation.  

We start with important properties of $\Geom(p)$ matrices for $\frac{1}{2}\leq p<1$. 
\begin{lemma}\label{lem:properties}
For a satisfiable \ptdl{} KB $\K=(\TO,\A)$ with two $\dd_p$-instances of $\Geom(p)$, $\frac{1}{2}\leq p<1$ in the ABox,
and any matrix $M$ satisfying Conditions~\eqref{eq:matrix-formula} and~\eqref{eq:matrix-sum} of $\K$, we have
\begin{enumerate}
\item\label{item:lem:<} if $p>\frac{1}{2}$, then, for any $k\in \nats$, the
set 
\begin{equation}\label{eq:L(k)}
L(k)=\{M(0,k),M(1,k),\dots, M(k,k),M(k,k-1),\dots M(k,0)\}
\end{equation} 
contains at least one non-zero element,
\item\label{item:lem:=} if $p=\frac{1}{2}$, then
there exists at most one $k\in \nats$ such that all elements $L(k)$ are zeroes.
\end{enumerate}
\end{lemma}
\begin{proof}
By Property~\ref{prop:second} of \Geom and the fact that, for any matrix $M$ of $\K$ 
satisfying~\eqref{eq:matrix-formula} and~\eqref{eq:matrix-sum}, the elements $M(k,l)$ and $M(l,k)$, for all $l>k$, 
are bounded with $\Geom(p)(l)$, the statement of item~\ref{item:lem:<} is trivial. 

Item~\ref{item:lem:=} follows from the fact $\Geom(\frac{1}{2})(i) = \sum_{j>i} \Geom(\frac{1}{2})(j)$ for all
$i\in \nats$. Since $\sum_{l\in\nats} M(k,l)=\frac{1}{2^{k+1}}$, if $L(k)$ contains only zeros for some $k$, i.e., 
$\sum_{0\leq l\leq k} M(k,l)=0$, then $M(k,k+1)={\frac{1}{2^{k+2}}}$, $M(k,k+2)={\frac{1}{2^{k+3}}}$, etc. Thus, for all 
$m>k$ the set $L(m)$ contains at least two positive elements.
\qed
\end{proof}
The following example confirms that satisfiability of a KB depends on the chosen parameter $p$.
\begin{example}
The \ptdl{} KB $\K=(\TO,\A)$ with $\A=\{\dd_{p}H(a), \dd_{p}T(a)\}$ and $\TO=\{T \equiv \neg H \}$ is satisfiable if 
$p=\frac{1}{2}$. The matrix in this case has the form
\[
M=
\begin{bmatrix}
	\times & \frac{1}{4} & \frac{1}{8} & \frac{1}{16} &\dots \\
    \frac{1}{4} & \times & & & \\
    \frac{1}{8} & & \times & & \\
    \frac{1}{16} & & & \times & \\
    \dots & & & & \dots          
\end{bmatrix},
\]
where the positions $\times$ correspond to unsatisfiable temporal KBs~\eqref{eq:matrix-formula}, and the matrix 
entries are trivially equal to $0$.

However, the same \ptdl{} KB with $p>\frac{1}{2}$ does not have any model, since $M(0,0)$ is unsatisfiable and the 
value of $\Geom(p)(0) > \sum_{j>0} \Geom(p)(j)$ cannot be spread on the rest part of the matrix.
\end{example}
With these basic properties we can develop an (infinite) iterative process of building a matrix for a given \ptdl{} KB 
$\K$. A finite $(\ell+1) \times (\ell+1)$ matrix $M_\ell$ is called \emph{partial}, for $\ell \in \nats$, if 
\begin{itemize}
\item for any $k_1,k_2\in [0,\dots,\ell]$, if $M(k_1,k_2)>0$ then the KB $(\TO,\A_{k_1,k_2})$, where 
$\A_{k_1,k_2}$ is defined by~\eqref{eq:matrix-formula}, is satisfiable; and
\item for any $k\in [0,\dots,\ell]$, $\sum_{0\leq \{k_1,k_2\}\setminus \{k\}\leq \ell}  M(k_1,k_2)=\Geom(p)(k)$.
\end{itemize} 
Clearly, if we can prove there is a partial matrix $M_\ell$ for all $\ell\in \nats$, we can conclude that \ptdl{} $\K$ is satisfiable.
%

\begin{definition}
A pair of matrix entries $M(i,\ell),M(\ell,j)$, for $i,j \leq \ell$,
is called \emph{chained} if  
there is an odd chain of elements, 
\begin{equation}\label{eq:chain}
\{M(i,\ell), M(i,k_1), M(m_1,k_1), M(m_1,k_2),\dots, M(m_{h},j), M(\ell,j)\},
\end{equation} 
where $m_c,k_c < \ell$, for all $1\leq c \leq h\in \nats$, such that, for every element of this chain $M(s,t)$, the temporal KB $(\TO,\A_{s,t})$ is consistent, and every even element is in the chain, e.g., $M(i,k_1), M(m_1,k_2),\dots, M(m_{h},j)$, is chained. 
Trivially, the diagonal element $M(\ell,\ell)$ is chained with itself, if $(\TO,\A_{\ell,\ell})$ is consistent. 
\end{definition}

An important property of chained elements is that increasing $M(i,\ell)$, $M(\ell,j)$, $i,j <\ell$, to some value can be compensated in the sense of preserving exact sums of columns and rows like in~\eqref{eq:matrix-sum} by decreasing the second and the penultimate element in~\eqref{eq:chain}, which are chained and (as we see later) can have non-zero value, then increasing the third and the antepenultimate, etc until decreasing an element in the middle of the chain. 

Now we are ready to prove that a pair of chained elements for every $0\leq k\leq \ell$ ensures the finite matrix $M_\ell$ is partial.
%
\begin{restatable}{lemma}{Lemmaprocess}\label{lem:process}
Any \ptdl{} KB $\K=(\TO,\A)$ with two $\dd_p$-instances of $\Geom(p)$, $\frac{1}{2}\leq p<1$, i.e., 
$\nxt^{n_1}\dd_p\theta_1(\av_1),\nxt^{n_2}\dd_p\theta_2(\av_2) \in \A$ is satisfiable iff there is a matrix $M$ such that either 
\begin{enumerate}
\item\label{lem:process:item1} for any $k\in \nats$, 
there is a chained pair in $L(k)$; or
\item\label{lem:process:item1/2} if $p=\frac{1}{2}$ and there is $k\in \nats$ with no chained elements in $L(k)$, 
then a submatrix $M_{k-1}$ is a partial matrix and $M(i,k)=M(k,i)=\frac{1}{2^{i+1}}$ for all $i>k$.
\end{enumerate} 
\end{restatable}
%
\begin{proof}
We prove Item~\ref{lem:process:item1} both directions by induction, obtaining a partial matrix for every $k \in \nats$. 
%
For the base, $k=0$, by Lemma~\ref{lem:properties}, there are three options:
\begin{enumerate}
\item $p\geq \frac{1}{2}$ and a temporal KB $\K_{0,0}=(\TO,\A_{0,0})$, for $\A_{0,0}$ defined by~\eqref{eq:matrix-formula},
is satisfiable. Thus, we set $M(0,0)=p$, there is a trivial chained pair in $L(0)$, and the finite matrix $M_0$ is partial. 
\item $p > \frac{1}{2}$ and $\K_{0,0}$ is unsatisfiable, then, by Lemma~\ref{lem:properties}(\ref{item:lem:<}) and Theorem~\ref{th:matrix}, the \ptdl{} KB $\K$ is unsatisfiable.
\item $p = \frac{1}{2}$ and $\K_{0,0}$ is unsatisfiable, then we set $M(0,0)=0$ and, by Lemma~\ref{lem:properties}(\ref{item:lem:=}) and Theorem~\ref{th:matrix}, the \ptdl{} KB $\K$ is consistent iff all temporal KBs $(\TO,\A_{0,m})$ and $(\TO,\A_{m,0})$, $m>0$, are satisfiable. One can see that the $(\TO,\A_{0,m})$, for all $m>0$, are satisfiable iff a \ptdl{} KB $(\TO,A_{0,*})$ with an only one $\dd_p$-data instance in the ABox  
$$
A_{0,*}=\nxt^{n_1} \theta_1(\av_1) \cup \nxt^{n_2} \neg\theta_2 (\av_2) \cup \nxt^{n_2+1} \dd_{\frac{1}{2}} \theta_2(\av_2)  \cup \A \setminus \{\bigcup_{i=1}^{2} \nxt^{n_i}\dd_{\frac{1}{2}} \theta_i(\av_i) \}, 
$$  
is satisfiable. Note that with one $\dd_p$-data instance, the type of distribution does not affect the satisfiability as long as it is complete. The satisfiability of $(\TO,\A_{m,0})$, $m>0$, can be checked in the same way. If both these \ptdl{} KBs are consistent, we can assign $M(m,0)=M(0,m)=\frac{1}{2^{m+1}}$ for all $m>0$, and the item~\ref{lem:process:item1/2} of this lemma holds.
\end{enumerate}

Let the statement be correct for some $(k-1)\in \nats$. Consider the $k$-th case: 
\begin{enumerate}
\item if $(\TO,\A_{k,k})$ is satisfiable and, thus, there is a trivial chain in $L(k)$, then we set $M(k,k)=(1-p)^k\cdot p$ and all the rest elements of $L(k)$ as $0$s. By induction hypothesis, $M_{k}$ is a partial matrix.  
\item for $p \geq \frac{1}{2}$, if there is a chained pair $M(i,k)$ and $M(k,j)$, $i,j<k$,  
then we assign $M(i,k)=M(k,j)=(1-p)^{k}\cdot p$ and all the rest elements of $L(k)$ to $0$s. By the property of chained pairs, we can keep the finite matrix $M_k$ partial by alternating summation and subtraction the value of $(1-p)^{k}\cdot p$ from the elements in the chain~\eqref{eq:chain}. If we need to subtract from a chained element $M(s,t)=0$, for, without losing generality, $s\leq t<k$, we split the value $(1-p)^{t}\cdot p$ in $L(t)$ into two parts, $(1-p)^{k}\cdot p$ and $((1-p)^{t}-(1-p)^{k})\cdot p$. Since the geometric distribution is strictly decreasing, these values are positive. We have two chained pairs for $L(t)$, the one we picked before and another with $M(s,t)$. By the procedure we describe in this item, for $L(t)$, $t < k$, we reassign the number $((1-p)^{t}-(1-p)^{k})\cdot p$ to the former pair and $(1-p)^{k}\cdot p$ to the latter. 

Otherwise, if $M(s,t)>0$, by property~\ref{prop:second} of the geometric distribution and our procedure, $M(s,t) > (1-p)^{k+1}\cdot p$.
\item if $p > \frac{1}{2}$ and there is no pair of chained elements $M(i,k)$ and $M(k,j)$, $i,j<k$, we can distinguish several reasons. If all matrix entries $M(g,k)$ or $M(k,g)$, $g<k$ correspond to unsatisfiable temporal KBs, it contradicts Lemma~\ref{lem:properties}(\ref{item:lem:<}), by which we conclude the \ptdl{} KB $\K$ is inconsistent.

The two elements $M(i,k)$ and $M(k,j)$ of satisfiable temporal KBs are unchained because, for each path between $M(i,k)$ and $M(k,j)$, there is a number $l\leq k$ and a satisfiable entry $M(l,m)$, $m<l$, such that there is no chained pair in $L(l)$ for $M(l,m)$ due to each possible path contains a matrix entry of unsatisfiable temporal KB. Indeed, if there are two elements $M(i,k)$ and $M(k,j)$ of satisfiable temporal KBs $\K_{i,k}$ and $\K_{k,j}$ and there is a connecting path of satisfiable elements, but $M(s,t)$ on its even position is not chained, then we can consider each possible path from $M(s,t)$ to an element in $L(\max(s,t))$ and apply the inductive reasoning to this new unchained pair. Since at each step $\max(s,t) < k$, for each possible path between original $M(i,k)$ and $M(k,j)$, we find an element in $L(l)$ with all paths gone through an inconsistent KB.

Also, speaking about column and row sums, an existence of a chain guarantees linear dependence of two chained elements. In effect, for~\eqref{eq:chain}, by variable elimination, we obtain   
\begin{align*}
\begin{cases}
M(i,k_1)=(1-p)^{i}\cdot p -\underline{M(i,\ell)}- \sum_{g\neq \ell} M(i,g), \\
M(m_1,k_1)=(1-p)^{k_1}\cdot p -M(i,k_1)- \sum_{g\neq i} M(g,k_1), \\
\dots \\
\underline{M(\ell,j)}=(1-p)^{j}\cdot p -M(m_{h},j)- \sum_{g\neq {m_h}} M(g,j).
\end{cases} 
\end{align*}
The temporal KBs of the elements on the left side have to be consistent in order to pass this linear dependence. By solving it with the rest equations, we have $\underline{M(\ell,j)}=\underline{M(i,\ell)}$.

As there is no possible chain between $M(i,k)$ and $M(k,j)$, there are two disjoint sets of row/column numbers, elements of which are linearly dependent on $M(i,k)$ and $M(k,j)$. The temporal KBs for elements in the intersection of these sets are inconsistent. When all elements from $L(k)$ are of satisfiable temporal KBs, but unchained, there are at least two disjoint sets, one for vertical elements in $L(k)$, another -- for horizontal.
Consider one of the sets, $\{i,i_1,i_2,\dots,i_u\}$, $u<k-2$, and $j\not\in\{i,i_1,i_2,\dots,i_u\}$. The sum equations with $M(i,k), M(i_1,k),\dots M(i_u,k)$ of satisfiable KBs have the following form:
\begin{align*}
\text{horizontal}
\begin{cases}
M(i,i)+M(i,i_1)+\dots+ M(i,i_u)+M(i,k)=(1-p)^i\cdot p, \\
M(i_1,i)+M(i_1,i_1)+\dots+ M(i_1,i_u)+M(i_1,k)=(1-p)^{i_1}\cdot p,\\
\dots \\
M(i_u,i)+M(i_u,i_1)+\dots+ M(i_u,i_u)+M(i_u,k)=(1-p)^{i_u}\cdot p,
\end{cases} \\
\text{and vertical:}
\begin{cases}
M(i,i)+M(i_1,i)+\dots+ M(i_u,i)=(1-p)^i\cdot p, \\
M(i,i_1)+M(i_1,i_1)+\dots+ M(i_u,i_1)=(1-p)^{i_1}\cdot p,\\
\dots \\
M(i,i_u)+M(i_1,i_u)+\dots+ M(i_u,i_u)=(1-p)^{i_u}\cdot p.
\end{cases}
\end{align*}
By variable elimination, we have $M(i,k)+M(i_1,k)+\dots +M(i_u,k)=0$. As each matrix entry is in $[0,1]$, the only solution is $M(i,k)=M(i_1,k)=\dots =M(i_u,k)=0$.
Therefore, for any elements of the matrix $M_{k-1}$, if there are satisfiable but not chained elements in $L(k)$, they are equal to zero. Hence, by Lemma~\ref{lem:properties}(\ref{item:lem:<}), the \ptdl{} KB $\K$ is inconsistent.

\item if $p = \frac{1}{2}$ and no possible chained pair in $L(k)$, 
then we set $L(k)$ as all $0$s and, by Lemma~\ref{lem:properties}(\ref{item:lem:=}) and Theorem~\ref{th:matrix}, the \ptdl{} KB $\K$ is consistent iff two \ptdl{} KBs $(\TO,A_{k,*})$ and $(\TO,A_{*,k})$ for 
\begin{align*}
\A_{k,*}=\bigcup_{0\leq i < k} \nxt^{n_1+i} \neg\theta_1 (\av_1) \cup \nxt^{n_1+k} \theta_1(\av_1) \cup \bigcup_{0\leq j \leq k} \nxt^{n_2+j} \neg\theta_2 (\av_2) \cup \\
\nxt^{n_2+k+1} \dd_{\frac{1}{2}} \theta_2(\av_2)  \cup \A \setminus \{\bigcup_{i=1}^{2} \nxt^{n_i}\dd_{\frac{1}{2}} \theta_i(\av_i) \}
\end{align*}
and symmetrical formula for $A_{*,k}$, are satisfiable. 
\end{enumerate}
The whole process comes down to the search of a pair of chained elements in $L(k)$. Letting the iteration to infinity, the \ptdl{} KB is satisfiable iff there is no step $k\in \nats$ such that the process stops.
\qed
\end{proof}
The matrix building process is shown in the following example.
\begin{example}
\label{ex:half}
Consider a \ptdl{} ABox $\{\dd_{p}H(a), \nxt\dd_{p}T(a)\}$, where $p=\frac{1}{2}$, with a TBox $\{T\equiv \neg H \}$. 
The matrix building process, according to the proof of Lemma~\ref{lem:process}, runs as follows:
\[M_0=
\begin{bmatrix}
    \frac{1}{2}
\end{bmatrix}, \quad
M_1=
\begin{bmatrix}
	\frac{1}{2} & 0 \\
    \times & \frac{1}{4}          
\end{bmatrix}, \quad
M_2=
\begin{bmatrix}
	\frac{3}{8} & 0 & \frac{1}{8} \\
    \times & \frac{1}{4} & 0 \\
    \frac{1}{8} & \times & \times \\         
\end{bmatrix}, \quad
M_3=
\begin{bmatrix}
	\frac{5}{16} & 0 & \frac{1}{8} & \frac{1}{16}\\
    \times & \frac{1}{4} & 0 & 0 \\
    \frac{1}{8} & \times & \times & \times \\    
    \frac{1}{16} & \times & \times & \times 
\end{bmatrix}, \dots
\]
The sign $\times$ denotes matrix entries of unsatisfiable temporal KBs with the corresponding ABoxes of the 
form~\eqref{eq:matrix-formula}, and these matrix entries are trivially equal to~$0$.
Iterating this process to infinity, the element $M(0,0)$, as the head of chained pairs, stays positive as 
$(\frac{1}{2}-\sum_{n\geq 3} \frac{1}{2^n})=\frac{1}{4}$. Thus, \eqref{eq:matrix-formula} and \eqref{eq:matrix-sum} 
hold and the \ptdl{} KB is consistent. 
\end{example} 
It is worth noting that the condition $\frac{1}{2}\leq p<1$ is crucial for the building process in 
Lemma~\ref{lem:process}.

\begin{example}
Now we let $p$ from Example~\ref{ex:half} be $\frac{1}{5}$. Despite the KBs for matrix entries remaining the same, 
and, for all $L(k)$, $k\geq 2$, we have a pair of chained elements, this \ptdl{} KB is unsatisfiable. 
\end{example} 
In the next subsection we demonstrate how to finitise the process in Lemma~\ref{lem:process}.


\subsection{Periodical Properties}
One can notice that, for a KB $\K$, there are constants $s=s(|\K|)$ and $p=p(|\K|)$ such that, starting from some 
$\ell>s$ time point, the formulas corresponding to the elements from $L(\ell)$ are equisatisfiable with some elements 
from $L(\ell+p)$ and, moreover, we show the following result.

\begin{restatable}{lemma}{Lemmaperiod}\label{lem:period}
For any satisfiable \ptdl{} KB $\K$ over two $\dd_p$-data instances, $\frac{1}{2}\leq p<1$, and any matrix $M$ 
satisfying~\eqref{eq:matrix-formula} and~\eqref{eq:matrix-sum}, 
there exist two integers $s,p < 2^{\TO(|\K|)}$ 
such that,
%
for any $\ell \geq s$ and any pair of chained elements $M(i,\ell)$ and $M(\ell,j)$, $i,j\leq \ell$, their translations $M(i',\ell+p)$ and $M(\ell+p,j')$ are also chained, where 
\begin{align}
i'=\begin{cases}
i, &\text{ if }i < s, \\
i+p, &\text{ otherwise, }
\end{cases} 
\qquad \text{ and } \qquad
j'=\begin{cases}
j, &\text{ if }j < s, \\
j+p, &\text{ otherwise.}
\end{cases} \label{eq:translation}
\end{align}
\end{restatable}
%
\begin{proof}
We proceed by induction on the length of the chain, $2k+1$. First, if $k=0$ and $i=j=\ell \geq s$, we can use the 
same reasoning as for the periodic property of an LTL formula. 

Given \ptdl{} KB $\K$, let $\A^{\downarrow}$ the ABox obtained by substituting both occurrences of 
$\nxt^{n_i}\dd_p \theta_i(\av_i)$, $i=1,2$, in $\A$ with the tautology $\nxt^{n_i} T_i(\av_i)$ by a fresh concept or role 
name $T_i$, for which we add an axiom to the ontology, 
$\TO^\downarrow=\TO\cup \{T_i\equiv (\theta_i\vee\neg \theta_i)\}$. Clearly, among models of 
$\K^\downarrow=(\TO^\downarrow,\A^\downarrow)$ there are all temporal models $\Imc$ of any probabilistic model 
$\Pmc=(\Imc,\mu)$ of $\K$. Then, since $\K^\downarrow$ and $\K_{\ell,\ell}=(\TO,\A_{\ell,\ell})$ are temporal KBs 
and belong to a proper sub-language of 
$T_{\U\Since}\DL_\textit{bool}^{\smash{\mathcal{N}}}$~\cite{DBLP:journals/tocl/ArtaleKRZ14}, we apply the
polynomially big (in the size of $\K$) translation 
from~\cite{DBLP:journals/tocl/ArtaleKRZ14,DBLP:conf/frocos/ArtaleKRZ09} to $\K^\downarrow$ and $\K_{\ell,\ell}$ 
which results in equisatisfiable LTL formulas $\Phi^\downarrow$, $\Phi_{\ell,\ell}$. The detailed reduction is 
described in Appendix~\ref{app:reduction}. 
Note that ABox instances are added to translations as conjuncts. Thus, if $\Phi_{\ell,\ell}$ and $\Phi_{\ell+p,\ell+p}$ 
are consistent, its models are also models of $\Phi^\downarrow$.

Let $|\text{sub}(\Phi^\downarrow)|$ be the number of subformulas of $\Phi^\downarrow$. It is known, for any propositional formula, that, the number of subformulas $|\text{sub}(\Phi^\downarrow)|$ is polynomially bounded on the length of $\Phi^\downarrow$. Therefore, it is finite. 

Consider a B\"uchi automaton that recognises the LTL formula $\Phi^\downarrow$. If $\Phi_{\ell,\ell}$ is consistent, then the automata accepts also all its models. Since $\ell\geq s=2^{|\text{sub}(\Phi^\downarrow)|}+1$ and the size of the automata is bound with $2^{|\text{sub}(\Phi^\downarrow)|}$, there is a cycle of states with propositional translations of $(\neg\theta_1)^*,(\neg \theta_2)^*$ of the size $p < 2^{|\text{sub}(\Phi^\downarrow)|}$. Thus, any model of $\Phi_{\ell,\ell}$ can be turned into a model of the propositional formula $\Phi_{\ell+p,\ell+p}$ by repeating the cycle one more time. 

%
%
%
We assume the statement holds for all pairs belonged to chains of lengths $\leq 2k+1$. Consider $M(i,\ell),M(\ell,j)$ with the chain 
$$
\{M(i,\ell), M(i,k_1), M(m_1,k_1), M(m_1,k_2),\dots, M(m_{k},j), M(\ell,j)\},
$$
of length $2k+3$. Without losing generality, we also assume that the chained elements on even places, $M(i,k_1), M(m_1,k_2),\dots M(m_{k},j)$, are belonged to their own chains of the lengths $\leq 2k+1$. If it is not the case, we can consider a longer chain of an element from $L(h)$, $h<\ell$. Since chains lay inside the $h\times h$ square, we can find a chain with its elements belonged to at most $(2k+1)$-length chains.

Since our translation~\eqref{eq:translation} depends on the position, the new chain 
$$\{M(i',\ell'), M(i',k_1'), M(m_1',k_1'), \dots, M(\ell',j')\}$$
of length $2k+3$ exists if the temporal KB $\K_{f',g'}$ satisfiable for every element $M(f',g')$ in the chain. For 
$M(f',g')$, we have the following cases:
\begin{itemize}
\item $f,g < s$ is trivial: $\K_{f,g}$ is satisfiable as $M(f,g)$ is in the original chain;
\item $f,g \geq s$, then we can apply a similar reasoning as in the basis case to prove the temporal KBs $\K_{f+p,g+p}=(\TO,\A_{f+p,g+p})$ is satisfiable;
\item $f \geq s > g$, again, we use an LTL reduction, $\Phi_{f,g}$. In the B\"uchi automata for $\Phi^\downarrow$, there is a cycle of the size $p < 2^{|\text{sub}(\Phi^\downarrow)|}$ where all states contain the propositional translation of $(\neg\theta_1)^*$. Also, as $g < s$, a state of this cycle can contain the positive propositional $\theta_2^*$. Thus, by repeating the cycle, a model for $\Phi_{f+p,g}$ is obtained, since the conjunct $\nxt^{n_2}(\bigwedge_{0\leq i<g} \nxt^i(\neg\theta_2)^*\wedge \nxt^g \theta_2^*)$ in $\Phi_{f+p,g}$ is satisfied at the first round of the cycle;
%
\item $g \geq s > f$, this case can be shown by a similar reasoning as in the previous item.
\end{itemize}
By induction hypothesis, translations of chained elements $M(f,g)$, for $f,g<\ell$, remain chained. Thus, $M(i',\ell+p)$ and $M(\ell+p,j')$ are indeed chained. 
\qed
\end{proof}
\begin{theorem}
Satisfiability of \ptdl{} KBs with $\Geom(p)$, $\frac{1}{2}\leq p<1$ can be decided in~\textsc{ExpSpace}.
\end{theorem}
\begin{proof}
To check if the conditions of Lemma~\ref{lem:process}(\ref{lem:process:item1}) hold, we ``guess'' the positions of 
unsatisfiable temporal KBs in the finite matrix $M_{s+p}$, where $s+p < 2^{f(\K)+1}$ and $f(\K)$ is a polynomial 
function. Then, for every $k<s+p$, among $L(k)$ we choose a pair of chained entries. If it is possible, 
Lemma~\ref{lem:period} ensures that the partial matrix $M_{s+p}$ can be extended to infinity.

If $p=\frac{1}{2}$ and there is a number $k\in \nats$ with no chained elements in $L(k)$ as in 
Lemma~\ref{lem:process}(\ref{lem:process:item1/2}), then $k<s+p$. Otherwise, if $k \geq s+p$, it contradicts 
Lemma~\ref{lem:period} which guarantees existence chained elements in $L(k)$, if they exist for all $L(i)$, $i<s+p$. 
Thus, we need to find a partial matrix $M_{k-1}$, which is in \textsc{NExpSpace} as in the case $\frac{1}{2}< p<1$, 
and then check the satisfiability of two exponentially big LTL formulas with an only one $\dd_p$-instance.
A \textsc{coNExpTime} oracle for verifying this has been presented in~\cite{ourKR}. 

Finally, in both cases, by Savitch's theorem, the satisfiability problem belongs to~\textsc{ExpSpace}. 
\qed
\end{proof}

The restriction to geometric distribution is not essential for most results. The main theorems hold for any (even not 
complete) distribution $\delta$ with the property $\delta(i)> \sum_{j>i}\delta(j)$, for all $i,j\in \dom(\delta)$. However, 
the case $\delta(i) = \sum_{j>i}\delta(j)$ relies on the procedure from~\cite{ourKR} which exploits complete 
distribution properties.

\section{Conclusions}
We have proposed a probabilistic temporal DL that is derived from a temporal \DL{} by specifying an exact 
distribution for a concept or a role in the data to be observed. We have also provided a first but substantial analysis 
of the complexity of reasoning in this logic, considering the standard reasoning task of KB consistency.
The importance of our formalism arises from the fact that temporal observations of events can usually be 
predicted with a probabilistic distribution over time; e.g., through an analysis of historical data.

This work is a first step towards a full formalism of uncertain temporal evolution of events, based on DLs.
Our work extends previous results~\cite{ourKR} developed for LTL, which can be seen as a special case of 
\ptdl{} where only one individual name exists. This paper provides new results, where more than one occurrence
of the distribution eventuality $\dd_\delta$ may be observed in a model.

Following the footsteps of~\cite{DBLP:conf/dlog/JungL13}, an interesting direction for future research is to consider 
query answering under temporal ontologies in data-centric applications with uncertain temporal data.
Along with a possible extension of temporal ontologies with interval-valued probabilistic constraints, for future work we also want to obtain effective methods for computing probabilities of different events and answer 
different types of probabilistic queries. Another possible line for research is the computation of the \emph{expected} 
(essentially, the average) time required until a desired property is observed, as it was previously done 
in~\cite{ourKR}.


\bibliographystyle{splncs03}
\bibliography{local}

\newpage
\appendix
%

\section{Reduction to First-Order Temporal Logic}\label{app:reduction}
In this section we apply the reduction provided in~\cite{DBLP:journals/tocl/ArtaleKRZ14,DBLP:conf/frocos/ArtaleKRZ09} from a temporal KB $\K$ with $\dd_\delta$-free ABox to $\mathcal{QTL}^1$, the one-variable fragment of first-order temporal logic over $(\nats,<)$, without existential quantifiers.

With every individual name $a\in \ind(\K)$ we associate the individual constant $a$ of $\mathcal{QTL}^1$, concept names $A$ to unary predicates $A(x)$, and existential concepts $\exists R$ to unary predicates $ER(x)$. Let
$\mathit{role}(\K)$ be the set of rigid and flexible role names, together with their inverses, occurring in $\K$.

By induction on the construction of a concept $C$, we define the $\mathcal{QTL}^1$-formula $C^*(x)$
\begin{align*}
A^*&=A(x),\ &&\bot^*=\bot,\ &(\exists R)^*=ER(x) \\
(C_1\U C_2)^*&=C_1^*\U C_2^*,\ &&(C_1\sqcap C_2)^*=C_1^*\wedge C_2^*,\ &(\neg C)^*=\neg C^*,
\end{align*}
and a similar translation for $\Diamond C$, $\nxt C$, $\Box C$. For the TBox $\T$,
$$
\T^\dagger= \bigwedge_{C_1 \sqsubseteq C_2 \in \T} \Box \forall x (C_1^*(x) \rightarrow C_2^*(x))
$$ 
and, for the RBox $\R$,
$$
\R^\dagger= \bigwedge_{\scriptsize\shortstack{$R_1 \sqsubseteq R_2 \in \R$ or \\ $R_1^- \sqsubseteq R_2^- \in \R$}} \Box \forall x (ER_1(x) \rightarrow ER_2(x))
$$ 

The following two properties for roles, if $R$ is a rigid role, then an $R$-successor at some moment implies an $R$-successor at all moments of time; if the domain of a role is not empty, then its range is not empty either, are encoded as 
\begin{align}
\bigwedge_{R \text{ is rigid}} \Box\forall x (\Diamond(\exists R)^*(x) \rightarrow \Box (\exists R)^*(x)), \label{eq:translation-role1} \\
\bigwedge_{R \in \mathit{role}(\K)} \Box \forall x((\exists R)^*(x) \rightarrow \exists x (\exists R^-)^*(x))\label{eq:translation-role2}
\end{align}
Formula~\eqref{eq:translation-role2} can be substituted with $\exists$-free expression, 
\begin{equation}\label{eq:translation-role2!}
\bigwedge_{R \in \mathit{role}(\K)} \Box \forall x(\Diamond(\exists R)^*(x) \rightarrow \Box p_R)\wedge(p_{R^-} \rightarrow (\exists R^-)^*(d_R)),
\end{equation} 
for a fresh constant $d_R$ and a fresh propositional variable $p_R$. The variables $p_R$ and $p_{R^-}$ indicate that $R$ and $R^-$ are non-empty whereas $d_R$ and $d_{R^-}$ witness that at 0.

Denote by $\sqsubseteq_\R^*$ the reflexive and transitive closure of the relation 
$$((R,R'),(R^-,(R')^-)\;|\; R \sqsubseteq R' \in \R),$$ where $R^-=S^-$, if $R=S$, and $R^-=S$, otherwise.

We assume that $A$ contains $\nxt^n R^-(b,a)$ whenever it contains $\nxt^n R(a,b)$. For each explicitly mentioned in the ABox $n\in \nats$ and each role $R$, we define the temporal slice $\A^R_n$ and $\A^R_\Box$ of $\A$ by taking
\begin{equation*}
\A^R_\Box=\begin{cases}
\{R(a,b) \;|\; \nxt^m R'(a,b)\in \A, m\in \nats, \text{ and }R'\sqsubseteq_\R^* R\}, &R \text{ is a rigid role} \\
\{R(a,b) \;|\; R'(a,b)\in \A^{R'}_\Box, \text{ for } R'\sqsubseteq_\R^* R\}, &R \text{ is a flexible role}.
\end{cases}
\end{equation*}
and 
\begin{equation*}
\A^R_n=\begin{cases}
\A^R_\Box, &R \text{ is a rigid role} \\
\{R(a,b) \;|\; \nxt^n R'(a,b)\in \A, \text{ for } R'\sqsubseteq_\R^* R \}  \\
\qquad\cup \{R(a,b) \;|\; R'(a,b)\in \A^{R'}_\Box, \text{ for } R'\sqsubseteq_\R^* R \} , &R \text{ is a flexible role}.
\end{cases}
\end{equation*}
We also set $\A_\Box=\bigcup_{R\in \mathit{role}(\K)}\A^R_\Box$ and $\A_n=\bigcup_{R \in \mathit{role}(\K)}\A^R_n$, for all $n\in \nats$ in the ABox.

The $\mathcal{QTL}^1$ translation of the ABox is defined as follows:
\begin{align*}
\A^\dagger=\bigwedge_{\nxt^n A(a)\in \A} \nxt^n A(a) \wedge \bigwedge_{\nxt^n\neg A(a)\in \A} \nxt^n \neg A(a)\wedge \bigwedge_{R(a,b)\in \A_n} \nxt^n (\exists R)^* (a)\wedge \\
\bigwedge_{R(a,b)\in \A_\Box} \Box (\exists R)^* (a)\wedge
\bigwedge_{\scriptsize\shortstack{$\nxt^n \neg R(a,b)\in \A$ \\ $R(a,b)\in \A_n$}} \bot.
\end{align*}
The $\mathcal{QTL}^1$ translation $\K^\dagger$ of $\K$ is the conjunction of $\T^\dagger$, $\R^\dagger$, $\A^\dagger$ and formulas~(\ref{eq:translation-role1},\ref{eq:translation-role2!}). The size of $\K^\dagger$ is polynomial in the size of $\K$.

\begin{lemma}[\cite{DBLP:journals/tocl/ArtaleKRZ14}]
A temporal KB $\K$ is satisfiable iff the $\mathcal{QTL}^1$ sentence $\K^\dagger$ is satisfiable.
\end{lemma}
Since $\K^\dagger$ contains no existential quantifiers, it can be regarded as a propositional temporal
formula because all the universally quantified variables can be instantiated by all the constants in
the formula with a polynomial blow-up.

\end{document}